%% file: acl_latex.tex
\documentclass[11pt]{article}

\usepackage[preprint]{acl}
\usepackage{enumitem}
\usepackage{times}
\usepackage{latexsym}
\usepackage[T1]{fontenc}
\usepackage[utf8]{inputenc}
\usepackage{microtype}
\usepackage{inconsolata}
\usepackage{graphicx}
\usepackage{booktabs}
\usepackage{float}
\usepackage{amsmath}
\usepackage{amssymb}
\usepackage{tikz}
\usetikzlibrary{arrows.meta,positioning,shapes.geometric,fit,calc,backgrounds}
\usepackage{amsmath,amssymb,amsthm}

\newtheorem{definition}{Definition}
\newtheorem{assumption}{Assumption}
\newtheorem{lemma}{Lemma}
\newtheorem{theorem}{Theorem}
\newtheorem{corollary}{Corollary}

\definecolor{syncblue}{HTML}{2C6E9B}
\definecolor{slored}{HTML}{C23B22}
\definecolor{ghostgray}{HTML}{9AA1A8}
\definecolor{pressline}{HTML}{5B6B7A}
\definecolor{goodgreen}{HTML}{2E7D4F}
\definecolor{darktext}{HTML}{222222}
\definecolor{secgray}{HTML}{6A6A6A}

\title{Not Every Sync Is Safe: Calibrated DiLoCo Scheduling for Shared AI Infrastructure}

\author{
 \textbf{ Maxwell Twelftree\textsuperscript{1}},
 \textbf{ David Lemphers\textsuperscript{1}},
 \textbf{An-chi He\textsuperscript{1}},
 \textbf{Yue Yang\textsuperscript{1}},
\\
\\
 \textsuperscript{1}Maincode, Melbourne, Australia,
\\
 \small{
   \textbf{Correspondence:} \href{mailto:yue@maincode}{yue@maincode}
 }
}

\begin{document}
\maketitle
\begin{abstract}
DiLoCo-style training reduces communication by letting learner islands train locally before occasional outer synchronization, making it attractive for fragmented industrial AI fleets where training shares hardware with latency-sensitive serving. The question for such fleets is when an outer merge is worth its system cost, and whether choosing \emph{which} windows to defer matters at all. Existing scheduling studies evaluate workload-aware policies against fixed-period baselines, but most omit the control that isolates timing from budget: matched random deferral, which inherits the controller's synchronization budget but is not itself deployable. This omission is consequential: across controlled stress tests and real vLLM sidecar replays, matched random ties or beats every forecast-free policy we test, so gains reported against weaker baselines cannot be attributed to window choice. We fill this gap with Workload-Aware DiLoCo (WA-DiLoCo), a score-based controller that weighs learner progress against fleet pressure, and a calibration protocol that determines when matched random can be beaten, then demonstrate that it can. In the bursty regime where calibration exposes request-overlap structure, adding a one-step EWMA burst forecast to the online controller beats matched random in real vLLM sidecar replay, reducing SLO violations from 6.54\% to 5.09\% (8 of 10 seeds, $p=0.021$); offline Calibrated-WA, a non-deployable bound, shows the remaining headroom at 4.45\% versus 6.26\%. The deployable lesson remains the protocol: report real-sidecar effect-size transfer, a no-sync load match, and a matched-random envelope before claiming serving-SLO improvement.
\end{abstract}
\section{Introduction}
Conventional distributed training assumes a tightly coupled cluster in which thousands of GPUs synchronize gradients at high frequency over expensive interconnects. Real industrial environments \vspace{-0.2em}
rarely satisfy this assumption: compute is fragmented across clusters, regions, cloud providers, hardware generations, and availability windows.

\begin{figure}[t]
\centering
\resizebox{\columnwidth}{!}{\input{fig_intuition_incl}}
\vspace{0.1em}
\caption{The scheduling problem. An outer merge occupies the fleet for $\sim$8\,s, and any request overlapping that window inherits the interference. Each row places the same budget of three merges: a fixed cadence collides with bursts, matched random (one draw) is equally blind, and a workload-aware policy defers merges into valleys. Placement is the only lever matched random does not capture; calibration (\S\ref{sec:eval-protocol}) decides when it is worth pulling.}
\label{fig:intuition}
\end{figure}
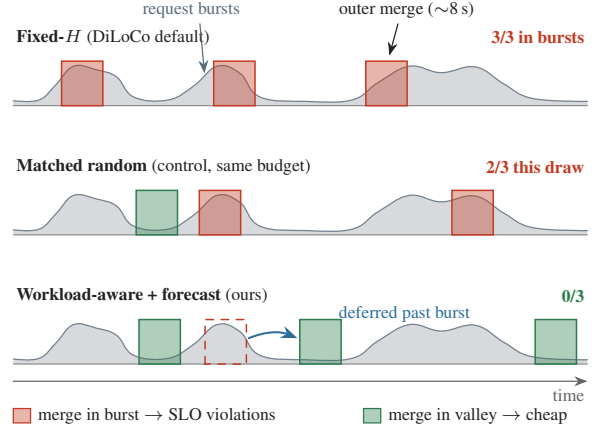
DiLoCo-style training reduces communication by allowing learner islands to perform many local optimization steps before an outer synchronization~\citep{douillard2023diloco}, which makes it attractive for these fragmented fleets. However, reducing how often synchronization occurs does not answer an equally important production question: when is it safe to synchronize? Recent DiLoCo variants make synchronization cheaper, less blocking, and more failure-tolerant~\citep{douillard2025streaming,douillard2026decoupled,koneputugodage2026factored,kolehmainen2025noloco,sarfi2025sparseloco}, but they do not decide whether the current serving state makes an outer merge safe.
This question becomes critical in shared fleets, where an outer merge can consume host, network, and checkpoint resources at exactly the moment latency-sensitive serving is under pressure; a fixed interval ignores fleet state, a pressure gate may defer progress too aggressively, and no single policy should be expected to dominate across all workload regimes.

To address this, we develop a calibration protocol for synchronization-scheduling claims, instantiated with Workload-Aware DiLoCo (WA-DiLoCo), a lightweight score-based controller used to test the protocol across regimes. The central object is not a universally dominant scheduler, but a decision procedure for when a training-side policy has earned a serving-SLO claim. The protocol's gate is matched random deferral: a non-deployable control that inherits the controller's synchronization budget post hoc and randomizes only placement (Figure~\ref{fig:intuition}). This separates how often to synchronize from which windows to choose. A scheduler that cannot beat its own budget placed at random has shown budget control, not workload awareness. Held to that bar, real vLLM replay yields a regime map. Pressure-heavy load is saturated; steady sync-heavy load lets WA-DiLoCo beat fixed and gated baselines but not matched random; bursty load exposes serving-cost structure that an offline calibrated variant exploits; and a pre-specified EWMA forecaster built from the same overlap signal beats matched random in real sidecar replay. The result is a protocol-first contribution: calibrate the serving mechanism, report the matched-random envelope, and deploy workload-aware synchronization only when placement has measurable value. We make four contributions:

\begin{itemize}[ leftmargin=*, topsep=1pt, itemsep=0pt, parsep=0pt ] \item \textbf{A calibrated evaluation protocol for shared-fleet DiLoCo.} We identify a gap in workload-aware synchronization evaluation: fixed-period baselines do not separate synchronization budget from window placement. Our protocol requires real-sidecar effect-size transfer, a no-sync load match, and a matched-random envelope before reporting serving-SLO gains. \item \textbf{A budget-matched random-deferral control.} We introduce matched random deferral as a non-deployable control that inherits the evaluated scheduler's merge budget and randomizes placement. It can perform well on average, but gives no decision-time guarantee for the realized run and cannot deliberately avoid high-cost windows. \item \textbf{A bounded calibrated scheduling formulation.} We formalize synchronization as feasible placement within a min/max-\(H\) corridor and show that, under a calibrated overlap model, offline calibrated scheduling is near-optimal within the matched-budget class, while the online variant degrades gracefully with forecast error. \item \textbf{An extensive real-sidecar regime study.} We run controlled stress tests and real vLLM sidecar replays across ordinary, pressure-heavy, steady sync-heavy, bursty, and non-periodic serving regimes. The experiments show when claims fail, when matched random is the right envelope, and when calibrated overlap plus forecasting improves SLOs. 
\end{itemize}

\section{Background}

\paragraph{DiLoCo-style training.}
DiLoCo partitions training into learner islands that perform many local inner steps, send model deltas to an outer synchronizer, and receive the redistributed global model \citep{douillard2023diloco}.
OpenDiLoCo provides an open implementation of this low-communication pattern \citep{du2024opendiloco}.

\paragraph{Communication and resilience.}
Streaming, gossip, decoupled, NoLoCo, and SparseLoCo variants reduce peak bandwidth, blocking, synchronization frequency, failure sensitivity, or pseudo-gradient payloads \citep{douillard2025streaming,douillard2026decoupled,koneputugodage2026factored,kolehmainen2025noloco,sarfi2025sparseloco}.
These methods improve communication and optimizer mechanics; we target the remaining scheduling problem when synchronization competes with production workloads.

\paragraph{Co-located training and serving.}
Recent systems schedule mixed training, fine-tuning, and inference workloads with SLO awareness, Pollux and Gandiva optimize cluster goodput and utilization \citep{chen2025mudi,li2025lemix,li2025mace,huang2026collm,oliaro2026flexllm,qiao2020pollux,xiao2018gandiva}, and vLLM improves serving latency through efficient KV-cache management and continuous batching \citep{kwon2023pagedattention}.
WA-DiLoCo is complementary: it exposes low-communication pretraining merges as schedulable events, then asks which deployment claim a measurement actually supports. Interference-prediction schedulers such as Paragon and Quasar~\citep{delimitrou2013paragon,delimitrou2014quasar} classify cross-workload interference for placement decisions; our protocol instead gates whether a training-side event may run now, against measured serving-side effect-size transfer.

\section{Methodology: Workload-Aware DiLoCo Synchronization}
\label{sec:method}

\subsection{Problem Setup}
\label{sec:setup}

In standard DiLoCo, learners perform many local inner updates and an outer synchronization, triggered at a fixed interval $H$, periodically merges their model updates. A fixed interval is not well matched to shared fleets. Figure~\ref{fig:intuition} illustrates the scheduling problem on a bursty serving trace: an outer merge occupies host, network, and checkpoint resources for $D=8$\,s in our configuration, and any latency-sensitive request in flight during that window inherits the interference. At matched tokens, short-run train loss is insensitive to placement within the min/max-$H$ budget (Table~\ref{tab:loss}), so synchronization becomes a one-sided scheduling problem: move a fixed budget of merges out of bursts and into valleys, but only when the trace has structure a deployable signal can see.

Let \(N\) learner islands train alongside a latency-sensitive serving workload
on the same fleet. A merge schedule is a sequence of outer-merge start times
\(S=(u_1,\ldots,u_K)\). Each outer merge has duration \(D\), which is
\(D=8\) seconds in our experiments. If the \(k\)-th merge starts at time
\(u_k\), then its merge window is
\[
W_k=[u_k,u_k+D].
\]
A serving request \(r\) has arrival time \(a_r\), completion time \(d_r\), and
request interval \(I_r=[a_r,d_r]\). We call request \(r\) \emph{sync-active}
under schedule \(S\) if it overlaps at least one outer-merge window:
\[
I_r\cap\bigcup_{k=1}^{K} W_k \neq \emptyset .
\]

Let \(\ell_r(S)\) denote the realized latency of request \(r\) under merge
schedule \(S\), and let \(\theta\) denote the latency SLO threshold for the
serving regime. A request violates the SLO if its realized latency exceeds
this threshold. Therefore, the SLO violation rate over request set \(R\)
measures the fraction of served requests that miss the latency target:
\[
V(S)=
\frac{1}{|R|}
\left|
\{r\in R:\ell_r(S)>\theta\}
\right|.
\]
When a load-matched no-synchronization run exists, we also report the
sync-excess violation rate
\[
V_{\mathrm{ex}}(S)=V(S)-V_{\mathrm{nosync}},
\]
where \(V_{\mathrm{nosync}}\) is measured under the same seed and request trace
without outer synchronization. This subtracts the base serving-load violation
rate and isolates avoidable SLO damage attributable to synchronization timing.

A schedule is the merge-start sequence $S=(u_1,\dots,u_K)$; let $g_k$ denote the gap in inner steps between consecutive merges, with $g_1$ measured from the horizon start. A schedule is \emph{feasible} when every gap respects the corridor
\begin{equation*}
H_{\min}\,\leq\, g_k \,\leq\, H_{\max}\qquad \text{for all } k,
\end{equation*}
and we write $\mathcal{F}$ for the feasible set over the training horizon and $\mathcal{F}_K\subset\mathcal{F}$ for the feasible schedules containing exactly $K$ merges. Every policy in this paper is then a decision-time heuristic for the one-sided program
\begin{equation*}
\min_{S\in\mathcal{F}}\; V_{\mathrm{ex}}(S),
\end{equation*}
and policies differ only in what they may observe when choosing $u_k$.

All policies share the same token-weighted outer merge. If learner $i$ processes $T_i$ tokens since the previous synchronization and proposes local update $\Delta_i$, the merged update is
\begin{equation*}
\Delta = \frac{\sum_{i=1}^{N} T_i\Delta_i}{\sum_{i=1}^{N} T_i}.
\end{equation*}
The merged model is redistributed before the next round, so policy differences are attributed to \emph{when} synchronization occurs rather than how updates are combined.

\subsection{WA-DiLoCo Scheduler}
\label{sec:scheduler}

WA-DiLoCo is a score-based controller for deciding when an outer merge should run. At each decision point, learner $i$ reports loss reduction $\Delta\ell_i\in\mathbb{R}_{\geq0}$, token contribution $T_i$, and staleness $A_i$. The coordinator also observes normalized network, serving, and checkpoint pressure $P_n,P_q,P_c\in[0,1]$. Heterogeneous learner signals are mapped to $[0,1]$ by the capped linear scale
\begin{equation*}
\mathrm{cap}(x;c)=\min(\max(x/c,0),1),
\end{equation*}
with $c_\ell=0.05$ for loss reduction, $c_t=10^6$ for tokens, and $c_a=1800$\,s for staleness. For learner $i$, WA-DiLoCo computes
\begin{equation*}
\begin{split}
s_i ={}& \lambda_{\ell}\,\mathrm{cap}(\Delta\ell_i;c_\ell)
       + \lambda_t\,\mathrm{cap}(T_i;c_t) \\
      &- \lambda_n P_n - \lambda_q P_q - \lambda_c P_c
       - \lambda_a\,\mathrm{cap}(A_i;c_a).
\end{split}
\end{equation*}
Progress terms raise the score; pressure and staleness lower it. The deployed weights are $(\lambda_\ell,\lambda_t,\lambda_n,\lambda_q,\lambda_c,\lambda_a)=(0.5,0.5,0.5,0.75,0.25,0.25)$; serving pressure carries the largest fleet-side weight because the target setting co-locates training with latency-sensitive serving. Under these weights, $s_i\in[-1.75,1]$: capped progress contributes at most $+1$, while maximal fleet pressure contributes $-1.5$ before staleness. Network and checkpoint terms remain harness signals in this study; no real-sidecar claim rests on them.

Let $h$ be the number of inner steps since the previous merge and let $\bar{s}=N^{-1}\sum_i s_i$. From $h\geq H_{\min}$ onward, the coordinator synchronizes iff
\begin{equation*}
\mathrm{sync}(h)=
\begin{cases}
1, & h\geq H_{\max},\\
1, & h\geq H_{\min}\ \text{and}\ \bar{s}>\tau,\\
0, & \text{otherwise},
\end{cases}
\end{equation*}
with $H_{\min}=128$, $H_{\max}=512$, and $\tau=0$. Thus the controller chooses placement within the bounded corridor $H_{\min}\leq g_k\leq H_{\max}$, not an unconstrained synchronization budget. The protocol in Section~\ref{sec:eval-protocol} audits which variants earn a serving-SLO claim.

\subsection{Calibrated Overlap Cost and Forecasting}
\label{sec:overlap}

For bursty serving, a more direct serving-cost signal exists than WA-DiLoCo's pressure terms: synchronization is most harmful when its window overlaps sidecar requests. For a candidate merge starting at $u$, define
\begin{equation*}
C(u)=\left|\{r:[a_r,d_r]\cap[u,u+D]\neq\emptyset\}\right|,
\end{equation*}
the number of sidecar requests in flight during the merge window. Bursty sidecar calibration fits a monotone map $\hat{\psi}$ on paired observations of schedule overlap and measured sync-excess violations, yielding the proxy

\begin{equation*}
\hat{V}_{\mathrm{ex}}(S)=\hat{\psi}\!\left(\textstyle\sum_{k} C(u_k)\right),
\end{equation*}
expressing schedule overlap in the units of $V_{\mathrm{ex}}$.

Offline Calibrated-WA keeps WA-DiLoCo's synchronization count but reselects the lowest-overlap feasible windows after observing the full burst trace; not deployable, it diagnoses whether exploitable low-cost windows exist. Online Calibrated-WA uses only decision-time expected overlap of the immediate next window, separating the value of the measured signal from the controller's ability to exploit it online. The EWMA forecaster adds the one-step request-rate estimate
\begin{equation*}
\hat{\lambda}_t=\alpha r_t+(1-\alpha)\hat{\lambda}_{t-1},
\end{equation*}
where $r_t$ is the recent request count, and scores a feasible window by
\begin{equation*}
\hat{C}(u)=q(u)+\hat{\lambda}_tD,
\end{equation*}
where $q(u)$ counts currently in-flight requests. At each decision step it commits to

\begin{equation*}
u^{*}\in\operatorname*{arg\,min}_{u\in\,\mathcal{U}_t}\hat{C}(u),
\end{equation*}
where $\mathcal{U}_t$ is the candidate set feasible under the $[H_{\min},H_{\max}]$ corridor, with a forced merge at $h=H_{\max}$ and the main controller's budget accounting.

\section{Evaluation Protocol}
\label{sec:eval-protocol}

Figure~\ref{fig:protocol} summarizes our protocol: a controlled stress harness tests whether synchronization timing can affect the SLO/progress frontier, and real-sidecar calibration determines whether the effect transfers to actual serving latency.

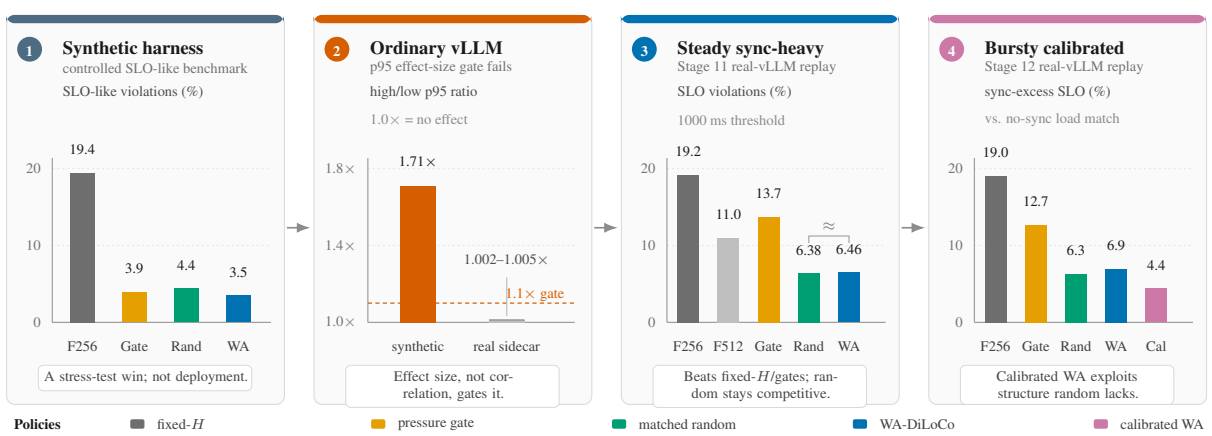
\begin{figure*}[t]
\centering
\input{fig1_incl}
\caption{Robustness matrix for DiLoCo synchronization policies. The stress harness shows that synchronization timing can affect the SLO/progress frontier, but ordinary vLLM sidecars show that synthetic pressure does not always transfer to real serving. Under steady sync-heavy replay, WA-DiLoCo outperforms fixed and gate baselines, although matched random deferral remains competitive. Under bursty serving, offline Calibrated-WA performs best by selecting synchronization windows with lower measured request overlap.
}
\label{fig:protocol}
\end{figure*}

\paragraph{Stress-harness test.}
The harness alternates low/high pressure windows, emits 32 synthetic requests every 5 seconds, applies a 120\,ms SLO, and adds an 8-second post-merge synchronization penalty; a window is high pressure when the maximum of the three pressure signals reaches 0.6. This stage tests controlled scheduling sensitivity, not production SLO improvement.

\paragraph{Matched baselines.}
We compare fixed-$H$ 256, fixed-$H$ 512, pressure gate, matched pressure gate, random deferral, matched random deferral, and WA-DiLoCo; the matched baselines ensure WA-DiLoCo cannot win merely by synchronizing less often or receiving more deferral opportunities. We treat matched random deferral as a gate in its own right: a serving-SLO claim must survive the matched-random envelope, not only fixed and gated baselines. Formally, for a policy $\pi$ whose realized schedule $S_\pi$ contains $K_\pi$ merges, the matched-random envelope is
\begin{equation*}
\mathrm{MR}(\pi)=\mathbb{E}_{S\sim\mathrm{Unif}(\mathcal{F}_{K_\pi})}\!\left[V(S)\right],
\end{equation*}
estimated per seed from post-hoc uniform draws, and the placement gain of $\pi$ is
\begin{equation*}
G(\pi)=\mathrm{MR}(\pi)-\mathbb{E}\!\left[V(S_\pi)\right].
\end{equation*}
$G(\pi)>0$ is exactly the claim that window choice, not sync rate, improves serving; $V_{\mathrm{ex}}$ replaces $V$ wherever a load-matched no-sync run exists.

\paragraph{Sidecar calibration.}
Stress-harness effects must transfer to real sidecars before supporting a production-serving claim. We replay pressure and synchronization traces through real sidecars and measure p95 effect-size transfer, not only pressure-latency correlation, summarized by the pressure ratio $\rho_p=\mathrm{p95}_{\mathrm{high}}/\mathrm{p95}_{\mathrm{low}}$ and the synchronization ratio $\rho_s=\mathrm{p95}_{\mathrm{sync\text{-}act}}/\mathrm{p95}_{\mathrm{inact}}$; without transfer in the relevant ratio, the result remains a controlled benchmark result.

\paragraph{Claim hierarchy.}
We separate five evidence tiers: controlled stress-harness behavior, real-sidecar effect-size transfer, policy replay inside a transferred regime, offline calibrated replay, and online controller behavior (Table~\ref{tab:claim_hierarchy}). Evidence at one tier motivates the next experiment but does not license the next claim; in particular, offline Calibrated-WA earns a best-in-regime replay claim only after calibration identifies synchronization-window interference, and a separate online experiment tests the same signal at decision time.

\paragraph{Isolation audit.}
Because shared-fleet state can diverge from scheduler-visible state, we audit isolation with a \texttt{nvidia-smi} guard, quarantine contaminated runs, and requeue clean runs before interpreting policy differences.

\section{Experiments}

\paragraph{Deployment motivation.}
Our experiments are motivated by a shared internal accelerator fleet rather than a dedicated DiLoCo cluster: the same pool runs pretraining experiments, vLLM sidecars, checkpoint I/O, and fabric tests, so an outer merge is a schedulable production event; we treat this study as deployment-motivated replay.
\subsection{Experimental Setup}

Experiments run on one 8$\times$B200 DGX node with one learner per GPU. The main training workload continues SmolLM2-135M on FineWeb-Edu using bf16, sequence length 1024, and AdamW; held-out evaluation uses WikiText-2. Real-serving calibration and replay use vLLM sidecars on the same B200 environment, with Qwen2.5 models and controlled pressure-heavy, sync-heavy, and bursty serving regimes.

For bursty replay, we compare WA-DiLoCo, Calibrated-WA, fixed-interval policies, pressure gates, matched random deferral, and a no-sync load-matched baseline. The reported sync-excess SLO subtracts the no-sync run under the same seed and burst trace, isolating avoidable SLO damage from synchronization timing. Headline comparisons appear in the main text (Table~\ref{tab:headline}); detailed calibration and policy replay tables, hyperparameters, seeds, sensitivity sweeps, statistical tests, and MI355X/RCCL fabric smoke-test details appear in Appendix~\ref{app:setup} and the released artifact.

\subsection{Results}
\label{sec:results}

\paragraph{Calibration Identifies the Valid Regimes}
Table~\ref{tab:calibration} is the first robustness check. The synthetic harness produces $\rho_p=1.71$--$1.72$, but ordinary Qwen2.5 vLLM sidecars from 0.5B to 7B remain nearly flat at $\rho_p=1.002$--$1.005$, even when pressure-latency correlation is high; the effect size does not transfer. Targeted sidecars expose two mechanisms: pressure-heavy serving reaches real $\rho_p=1.546$ (correlation 0.898), while sync-heavy serving is flat on pressure but reaches $\rho_s=1.556$, so synchronization-window interference transfers independently of pressure.
\vspace{-0.5em}

\paragraph{Policy Robustness Across Real vLLM Regimes}
Table~\ref{tab:headline} reports each policy as placement gain $G$ over its matched-random envelope (Section~\ref{sec:eval-protocol}); Appendix Table~\ref{tab:stage11} gives the full replay table. Policy choice is regime-dependent. Pressure-heavy replay transfers but is saturated, with all policies between $V=35.91$\% and $37.21$\%. Sync-heavy replay is controllable: WA-DiLoCo reaches $V=6.46$\%, beating fixed-$H$ 256, fixed-$H$ 512, pressure gate, and matched pressure gate on all five paired seeds ($p=0.031$), but matched random reaches 6.38\% and ties it statistically ($G=-0.08$pp, $p=0.56$), so steady load supports no dominance claim. Bursty sync-heavy replay is where calibration pays: offline Calibrated-WA reduces $V_{\mathrm{ex}}$ to 4.45\%, versus 6.92\% for WA-DiLoCo and an $\mathrm{MR}$ envelope of 6.26\% across 25 draws, winning all five seeds ($G=+1.81$pp, $p=0.031$). Online Calibrated-WA reaches $V_{\mathrm{ex}}=5.92 \pm 1.48$\% but still ties matched random ($G=+0.34$pp, $p=0.344$), leaving forecasting as the gap.

\noindent\textbf{Forecast, Harness, and Ablations}
The offline and online calibrated replays separate signal validity from online exploitability. Online Calibrated-WA validates the overlap-cost signal but does not significantly beat matched random by itself (G = +0.34pp, p = 0.344). We therefore test one pre-specified forecast-augmented variant: an EWMA request-rate forecaster that scores feasible windows by expected overlap before the merge is launched. On the overlap-calibrated proxy it wins 8 of 10 seeds against matched random (p = 0.020), reducing busy-window synchronization fraction from 0.242 to 0.125. Replaying the same schedules through real Qwen2.5-1.5B vLLM sidecars at a matched budget of eight synchronization events, the forecaster cuts SLO violations from 6.54 ± 1.60\% to 5.09 ± 1.19\%, a +1.46pp placement gain, wins 8 of 10 seeds (p = 0.021), and reduces sync-active requests from 286 ± 67 to 220 ± 52 (Tables 1 and 10).

\begin{table}[t]
\centering
\footnotesize
\setlength{\tabcolsep}{3.5pt}
\begin{tabular}{@{}lrrr@{}}
\toprule
\textbf{Policy} & \textbf{Viol.\ (\%)} & \textbf{$G$ (pp)} & \textbf{$p$} \\
\midrule
\multicolumn{4}{@{}l}{\emph{Steady sync-heavy, real sidecar ($V$, 5 seeds)}} \\
Random matched & $6.38 \pm 0.71$ & -- & -- \\
WA-DiLoCo & $6.46 \pm 0.35$ & $-0.08$ & $0.56$ \\
\midrule
\multicolumn{4}{@{}l}{\emph{Bursty replay ($V_{\mathrm{ex}}$, 5 seeds, 25-draw envelope)}} \\
Random matched (MR) & $6.26$ & -- & -- \\
WA-DiLoCo & $6.92$ & $-0.66$ & -- \\
Online Cal-WA & $5.92 \pm 1.48$ & $+0.34$ & $0.344$ \\
Offline Cal-WA$^{\dagger}$ & $4.45$ & $+1.81$ & $0.031$ \\
\midrule
\multicolumn{4}{@{}l}{\emph{Non-periodic, real sidecar ($V$, 10 seeds)}} \\
Random matched & $6.54 \pm 1.60$ & -- & -- \\
EWMA forecaster & $\mathbf{5.09 \pm 1.19}$ & $+1.46$ & $0.021$ \\
\bottomrule
\end{tabular}
\caption{Headline results by regime: violation rate and placement gain $G$ over the matched-random envelope (Section~\ref{sec:eval-protocol}); one-sided paired sign-flip $p$ versus matched random. $^{\dagger}$Offline bound, not deployable.}
\label{tab:headline}
\end{table}

\section{Discussion and Conclusion} \begin{table}[t] \centering \footnotesize \setlength{\tabcolsep}{4pt} \begin{tabular}{@{}lll@{}} \toprule \textbf{Regime} & \textbf{Evidence} & \textbf{Use} \\ \midrule No p95 transfer & H/L p95 $\approx$ 1.00 & Any policy \\ Pressure-heavy & All within 1.3pp & None; saturated \\ Steady sync-heavy & WA ties random & WA, but control ties it \\ Bursty calibrated & Cal-WA wins & Calibrated policy \\ Uncalibrated & -- & Run protocol first \\ \bottomrule \end{tabular} \caption{Policy selection by calibrated regime.} \label{tab:decision} \end{table}
\vspace{-0.5em}
\noindent\textbf{Operational reading.}
At the forecaster replay load, about 4{,}300 requests per 600\,s, the 1.46pp
gain corresponds to roughly 60 fewer SLO violations per sidecar per ten minutes
at identical training throughput. Calibration costs about one node-day of
sidecar replays and amortizes until the serving model, fleet, or merge payload
changes.

\noindent\textbf{Why matched random is a control, not a scheduler.}
Matched random is a strong envelope because it inherits the evaluated policy's
synchronization budget and randomizes only placement. It can perform well on
average, but a fleet executes one realized schedule: random placement gives no
decision-time guarantee that the next merge avoids a high-cost serving window.
Calibrated scheduling instead optimizes a measured overlap proxy inside the
feasible min/max-\(H\) corridor. Appendix~A formalizes the bound: under the
calibrated overlap model, the offline calibrated schedule is within
\(2\epsilon_{\mathrm{cal}}\) of the best feasible matched-budget schedule, while
the online version adds \(2L_{\psi}E_f\) forecast error. Thus matched random is
the right control, while calibrated scheduling is the bounded solution when the
overlap model is valid.

\noindent \textbf{Conclusion.}
WA-DiLoCo turns DiLoCo synchronization into an industry-facing control problem:
when shared AI fleets mix training, serving, checkpointing, and fabric pressure,
outer merges must be scheduled, not merely made cheaper. Our calibration protocol
sets the production bar: beat matched random in a calibrated deployment regime
before claiming serving-SLO impact. The EWMA controller meets that bar in
real-sidecar replay.

\section*{Limitations}
Our evaluation is a stress-tested replay study rather than a live production deployment. The policy-replay gates use real vLLM sidecars and real or replayed synchronization traces, but they are not live co-training deployments with customer traffic. The bursty replay demonstrates that serving-cost calibration can beat matched random deferral in that regime, not that the same schedule is universally optimal, and the online Calibrated-WA follow-up does not significantly beat matched random deferral, so the current online controller remains an incomplete implementation of the calibrated scheduling idea. The 8-second synchronization window reflects this configuration; merge duration scales with model and payload size, so calibration must be repeated per deployment. The forecast-augmented online result now includes real vLLM sidecar replay, but on one sidecar model (Qwen2.5-1.5B-Instruct), one non-periodic trace family, and ten seeds; it is replay, not live production deployment. The no-sync runs on this trace are not load-matched, so we report the forecaster comparison as head-to-head SLO at a matched synchronization budget rather than sync-excess SLO.

The observed regimes should also be interpreted as deployment-specific. Ordinary sidecars may appear insensitive because of B200 headroom, the tested model sizes, concurrency range, and request mix. However, this does not weaken the evaluation protocol; rather, it reinforces the need to stress-test synchronization policies under the target serving regime before making production-SLO claims. Training scale is another limitation. The main experiments use SmolLM2-135M on a single 8$\times$B200 node, with a shorter SmolLM2-1.7B scale check. These runs exercise real synchronization logic and accelerator-side behavior, but they do not capture full large-model optimizer dynamics or multi-node communication contention. Similarly, the MI355X/RCCL experiment validates outer-sync-sized collectives on real fabric, but it is not a full multi-node WA-DiLoCo policy replay. Finally, some controller inputs remain system signals rather than fully validated production mechanisms. Serving-side effects are tested through real sidecars, but network and checkpoint pressure are not yet evaluated through full fabric-contention or checkpoint-overlap policy replays. Because WA-DiLoCo consumes these signals, future deployments should tune weights from real fleet logs and add guardrails for starvation, fairness, rollback, and checkpoint safety.

\bibliography{custom}

\clearpage
\onecolumn
\appendix
\raggedbottom
\setlength{\intextsep}{9pt}
\setlength{\textfloatsep}{9pt}
\setlength{\floatsep}{9pt}
\setlength{\abovecaptionskip}{6pt}

\section{Theoretical Analysis}
\label{sec:theory}

This section gives a calibrated scheduling guarantee for WA-DiLoCo. We show that, under a
calibrated overlap model, offline calibrated scheduling is near-optimal among
feasible schedules with the same merge budget, while the online version degrades
gracefully with forecast error.

\subsection{Notation and Definitions}

\begin{definition}[Merge schedule]
Let \(D>0\) denote the duration of one outer synchronization. A schedule with
\(K\) outer merges is an ordered sequence
\[
S=(u_1,\ldots,u_K),
\]
where \(u_k\) is the start time of the \(k\)-th outer merge. The corresponding
merge window is
\[
W(u_k)=[u_k,u_k+D].
\]
\end{definition}

\begin{definition}[Synchronization corridor]
Let \(g_k(S)\) be the number of local inner steps between the \((k-1)\)-st and
\(k\)-th outer merge, with \(g_1(S)\) measured from the beginning of the
training horizon. Given two constants \(H_{\min}\) and \(H_{\max}\), a schedule
is feasible if
\[
H_{\min}\le g_k(S)\le H_{\max}
\qquad
\text{for all } k=1,\ldots,K.
\]
We write \(\mathcal{F}_K\) for the set of feasible schedules with exactly
\(K\) outer merges.
\end{definition}

\begin{definition}[Serving requests and SLO violations]
Let \(R\) be the set of serving requests in the trace. For each request
\(r\in R\), let \(a_r\) be its arrival time and \(d_r\) be its completion time.
Define the request interval
\[
I_r=[a_r,d_r].
\]
For a schedule \(S\), let \(\ell_r(S)\) be the realized latency of request
\(r\). Let \(\theta\) be the SLO latency threshold. The SLO violation rate under
schedule \(S\) is
\[
V(S)
=
\frac{1}{|R|}
\sum_{r\in R}
\mathbf{1}\{\ell_r(S)>\theta\}.
\]
Let \(V_{\mathrm{nosync}}\) be the violation rate in a load-matched run with no
outer synchronization. The sync-excess violation rate is
\[
V_{\mathrm{ex}}(S)
=
V(S)-V_{\mathrm{nosync}}.
\]
\end{definition}

\begin{definition}[Overlap cost]
For a candidate merge starting at time \(u\), define its request-overlap cost as
\[
C(u)
=
\sum_{r\in R}
\mathbf{1}\{I_r\cap [u,u+D]\neq \emptyset\}.
\]
Thus \(C(u)\) counts how many serving requests overlap the synchronization
window. For a full schedule \(S=(u_1,\ldots,u_K)\), define the total overlap
cost
\[
\Phi(S)
=
\sum_{k=1}^{K} C(u_k).
\]
\end{definition}

\subsection{Assumptions}

\begin{assumption}[Matched merge budget]
\label{assump:matched-budget}
All schedules compared in the analysis belong to the same feasible class
\(\mathcal{F}_K\). In particular, the fixed-interval DiLoCo schedule
\(S_H\), the offline calibrated schedule \(S_{\mathrm{off}}\), and the online
schedule \(S_{\mathrm{on}}\) all contain \(K\) outer merges and satisfy
\[
H_{\min}\le g_k(S)\le H_{\max}.
\]
\end{assumption}

\begin{assumption}[Calibrated overlap model]
\label{assump:calibration}
There exists a monotone nondecreasing calibration function
\[
\psi:\mathbb{R}_{\ge 0}\rightarrow \mathbb{R}
\]
and a calibration error \(\epsilon_{\mathrm{cal}}\ge 0\) such that, for every
feasible \(K\)-merge schedule \(S\in\mathcal{F}_K\),
\[
\left|
V_{\mathrm{ex}}(S)-\psi(\Phi(S))
\right|
\le
\epsilon_{\mathrm{cal}}.
\]
This assumption intentionally restricts the guarantee to regimes where total
request overlap is a calibrated proxy for synchronization-induced SLO damage.
The theory below should therefore be read as a calibrated-regime guarantee,
not an unconditional statement about all serving traces.
\end{assumption}

\begin{assumption}[Lipschitz calibration]
\label{assump:lipschitz}
For the online guarantee, we additionally assume that the calibration function
\(\psi\) is \(L_\psi\)-Lipschitz:
\[
|\psi(x)-\psi(y)|
\le
L_\psi |x-y|
\qquad
\text{for all } x,y\ge 0.
\]
\end{assumption}

Assumption~\ref{assump:calibration} is the key modeling assumption. It says
that total overlap \(\Phi(S)\) predicts sync-excess SLO damage up to error
\(\epsilon_{\mathrm{cal}}\). The theory below should therefore be interpreted
as a calibrated scheduling guarantee, not as an unconditional claim about every
possible serving trace.

\subsection{A Calibration Transfer Lemma}

\begin{lemma}[Calibration transfer]
\label{lem:calibration-transfer}
Under Assumption~\ref{assump:calibration}, for any two feasible schedules
\(S_A,S_B\in\mathcal{F}_K\),
\[
V_{\mathrm{ex}}(S_B)-V_{\mathrm{ex}}(S_A)
\ge
\psi(\Phi(S_B))-\psi(\Phi(S_A))-2\epsilon_{\mathrm{cal}}.
\]
Equivalently,
\[
V_{\mathrm{ex}}(S_A)
\le
V_{\mathrm{ex}}(S_B)
+
\psi(\Phi(S_A))-\psi(\Phi(S_B))
+
2\epsilon_{\mathrm{cal}}.
\]
\end{lemma}

\begin{proof}
By Assumption~\ref{assump:calibration},
\[
V_{\mathrm{ex}}(S_B)
\ge
\psi(\Phi(S_B))-\epsilon_{\mathrm{cal}},
\]
and
\[
V_{\mathrm{ex}}(S_A)
\le
\psi(\Phi(S_A))+\epsilon_{\mathrm{cal}}.
\]
Subtracting the second inequality from the first gives
\[
V_{\mathrm{ex}}(S_B)-V_{\mathrm{ex}}(S_A)
\ge
\psi(\Phi(S_B))-\psi(\Phi(S_A))-2\epsilon_{\mathrm{cal}}.
\]
This proves the claim.
\end{proof}

\subsection{Offline Calibrated-WA}

Offline Calibrated-WA is an oracle diagnostic. It observes the full serving
trace and chooses the feasible schedule with minimum total overlap:
\[
S_{\mathrm{off}}
\in
\operatorname*{arg\,min}_{S\in\mathcal{F}_K}
\Phi(S).
\]
This schedule is not deployable because it uses future request information. Its
purpose is to measure how much avoidable synchronization damage exists in the
trace.

Let
\[
S^\star
\in
\operatorname*{arg\,min}_{S\in\mathcal{F}_K}
V_{\mathrm{ex}}(S)
\]
be the best feasible \(K\)-merge schedule in terms of true sync-excess SLO
damage.

\begin{theorem}[Offline calibrated near-optimality]
\label{thm:offline-near-optimality}
Under Assumptions~\ref{assump:matched-budget} and
\ref{assump:calibration},
\[
V_{\mathrm{ex}}(S_{\mathrm{off}})
\le
V_{\mathrm{ex}}(S^\star)
+
2\epsilon_{\mathrm{cal}}.
\]
\end{theorem}

\begin{proof}
Since \(S_{\mathrm{off}}\) minimizes total overlap over \(\mathcal{F}_K\),
\[
\Phi(S_{\mathrm{off}})
\le
\Phi(S^\star).
\]
Because \(\psi\) is monotone nondecreasing,
\[
\psi(\Phi(S_{\mathrm{off}}))
\le
\psi(\Phi(S^\star)).
\]
Using Assumption~\ref{assump:calibration},
\[
V_{\mathrm{ex}}(S_{\mathrm{off}})
\le
\psi(\Phi(S_{\mathrm{off}}))+\epsilon_{\mathrm{cal}}.
\]
Combining the previous two inequalities gives
\[
V_{\mathrm{ex}}(S_{\mathrm{off}})
\le
\psi(\Phi(S^\star))+\epsilon_{\mathrm{cal}}.
\]
Again by Assumption~\ref{assump:calibration},
\[
\psi(\Phi(S^\star))
\le
V_{\mathrm{ex}}(S^\star)+\epsilon_{\mathrm{cal}}.
\]
Therefore,
\[
V_{\mathrm{ex}}(S_{\mathrm{off}})
\le
V_{\mathrm{ex}}(S^\star)+2\epsilon_{\mathrm{cal}}.
\]
\end{proof}

\begin{corollary}[Offline comparison with fixed-interval DiLoCo]
\label{cor:offline-fixed}
Let \(S_H\in\mathcal{F}_K\) be the fixed-interval DiLoCo schedule with the same
merge budget. Under Assumptions~\ref{assump:matched-budget} and
\ref{assump:calibration},
\[
V_{\mathrm{ex}}(S_H)-V_{\mathrm{ex}}(S_{\mathrm{off}})
\ge
\psi(\Phi(S_H))-\psi(\Phi(S_{\mathrm{off}}))
-
2\epsilon_{\mathrm{cal}}.
\]
In particular, if \(\psi(z)=\beta z\) for some \(\beta>0\), then
\[
V_{\mathrm{ex}}(S_H)-V_{\mathrm{ex}}(S_{\mathrm{off}})
\ge
\beta
\left[
\Phi(S_H)-\Phi(S_{\mathrm{off}})
\right]
-
2\epsilon_{\mathrm{cal}}.
\]
Thus \(S_{\mathrm{off}}\) strictly improves over fixed-interval synchronization
whenever
\[
\beta
\left[
\Phi(S_H)-\Phi(S_{\mathrm{off}})
\right]
>
2\epsilon_{\mathrm{cal}}.
\]
\end{corollary}

\begin{proof}
The first inequality follows directly from
Lemma~\ref{lem:calibration-transfer} with \(S_A=S_{\mathrm{off}}\) and
\(S_B=S_H\). The linear case follows by substituting \(\psi(z)=\beta z\).
Strict improvement holds whenever the lower bound is positive.
\end{proof}

\subsection{Online Calibrated-WA}

The online controller cannot observe the future serving trace. At each
decision, it estimates the overlap cost of feasible candidate merge windows
using current in-flight requests and an EWMA request-rate forecast.

\begin{definition}[Online candidate set]
Let \(j\in\{1,\ldots,K\}\) index online synchronization decisions. At decision
\(j\), let
\[
\mathcal{U}_j
\]
be the set of candidate merge start times that are feasible under the
\([H_{\min},H_{\max}]\) corridor. The set \(\mathcal{U}_j\) includes the forced
merge at \(H_{\max}\).
\end{definition}

\begin{definition}[True and predicted online overlap]
For candidate \(u\in\mathcal{U}_j\), let \(C_j(u)\) be the true overlap cost of
starting a merge at \(u\) from decision state \(j\). In the trace-based setting,
\(C_j(u)=C(u)\), but the subscript \(j\) emphasizes that the candidate is
evaluated at a particular online decision.

Let \(\hat C_j(u)\) be the online prediction of \(C_j(u)\). Online
Calibrated-WA chooses
\[
\hat u_j
\in
\operatorname*{arg\,min}_{u\in\mathcal{U}_j}
\hat C_j(u).
\]
The one-step clairvoyant choice at the same decision state is
\[
u_j^\circ
\in
\operatorname*{arg\,min}_{u\in\mathcal{U}_j}
C_j(u).
\]
\end{definition}

\begin{assumption}[Online forecast error]
\label{assump:forecast-error}
At every online decision \(j\), the predicted overlap error is bounded by
\(\epsilon_{f,j}\ge 0\):
\[
|\hat C_j(u)-C_j(u)|
\le
\epsilon_{f,j}
\qquad
\text{for all } u\in\mathcal{U}_j.
\]
Define the cumulative forecast error
\[
E_f
=
\sum_{j=1}^{K}\epsilon_{f,j}.
\]
If the error is uniformly bounded, we write
\[
\epsilon_{f,j}\le \epsilon_f
\qquad
\text{for all } j,
\]
so that \(E_f\le K\epsilon_f\).
\end{assumption}

\begin{lemma}[One-step online overlap regret]
\label{lem:one-step-regret}
Under Assumption~\ref{assump:forecast-error}, at every decision \(j\),
\[
C_j(\hat u_j)
\le
C_j(u_j^\circ)+2\epsilon_{f,j}.
\]
Consequently,
\[
\sum_{j=1}^{K} C_j(\hat u_j)
\le
\sum_{j=1}^{K} C_j(u_j^\circ)
+
2E_f.
\]
\end{lemma}

\begin{proof}
By Assumption~\ref{assump:forecast-error},
\[
C_j(\hat u_j)
\le
\hat C_j(\hat u_j)+\epsilon_{f,j}.
\]
Since \(\hat u_j\) minimizes predicted overlap over \(\mathcal{U}_j\),
\[
\hat C_j(\hat u_j)
\le
\hat C_j(u_j^\circ).
\]
Again by Assumption~\ref{assump:forecast-error},
\[
\hat C_j(u_j^\circ)
\le
C_j(u_j^\circ)+\epsilon_{f,j}.
\]
Combining the three inequalities yields
\[
C_j(\hat u_j)
\le
C_j(u_j^\circ)+2\epsilon_{f,j}.
\]
Summing over \(j=1,\ldots,K\) gives the second claim.
\end{proof}

Let \(S_{\mathrm{on}}\) be the schedule produced by Online Calibrated-WA:
\[
S_{\mathrm{on}}=(\hat u_1,\ldots,\hat u_K).
\]
Then
\[
\Phi(S_{\mathrm{on}})
=
\sum_{j=1}^{K} C_j(\hat u_j).
\]
Define the cumulative one-step clairvoyant overlap as
\[
\Phi_{\mathrm{myo}}^\circ
=
\sum_{j=1}^{K} C_j(u_j^\circ).
\]
Lemma~\ref{lem:one-step-regret} implies
\[
\Phi(S_{\mathrm{on}})
\le
\Phi_{\mathrm{myo}}^\circ+2E_f.
\]

\begin{theorem}[Online calibrated guarantee]
\label{thm:online-guarantee}
Suppose Assumptions~\ref{assump:matched-budget},
\ref{assump:calibration}, \ref{assump:lipschitz}, and
\ref{assump:forecast-error} hold. Suppose further that the sequence of
one-step clairvoyant choices defines a feasible comparator schedule
\[
S_{\mathrm{myo}}^\circ\in\mathcal{F}_K
\]
with
\[
\Phi(S_{\mathrm{myo}}^\circ)=\Phi_{\mathrm{myo}}^\circ.
\]
Then
\[
V_{\mathrm{ex}}(S_{\mathrm{on}})
\le
V_{\mathrm{ex}}(S_{\mathrm{myo}}^\circ)
+
2L_\psi E_f
+
2\epsilon_{\mathrm{cal}}.
\]
If \(\epsilon_{f,j}\le\epsilon_f\) for all \(j\), then
\[
V_{\mathrm{ex}}(S_{\mathrm{on}})
\le
V_{\mathrm{ex}}(S_{\mathrm{myo}}^\circ)
+
2L_\psi K\epsilon_f
+
2\epsilon_{\mathrm{cal}}.
\]
\end{theorem}

\begin{proof}
From Lemma~\ref{lem:one-step-regret},
\[
\Phi(S_{\mathrm{on}})
\le
\Phi(S_{\mathrm{myo}}^\circ)+2E_f.
\]
By Assumption~\ref{assump:lipschitz},
\[
\psi(\Phi(S_{\mathrm{on}}))
\le
\psi(\Phi(S_{\mathrm{myo}}^\circ))
+
2L_\psi E_f.
\]
Using Assumption~\ref{assump:calibration},
\[
V_{\mathrm{ex}}(S_{\mathrm{on}})
\le
\psi(\Phi(S_{\mathrm{on}}))+\epsilon_{\mathrm{cal}}.
\]
Therefore,
\[
V_{\mathrm{ex}}(S_{\mathrm{on}})
\le
\psi(\Phi(S_{\mathrm{myo}}^\circ))
+
2L_\psi E_f
+
\epsilon_{\mathrm{cal}}.
\]
Again by Assumption~\ref{assump:calibration},
\[
\psi(\Phi(S_{\mathrm{myo}}^\circ))
\le
V_{\mathrm{ex}}(S_{\mathrm{myo}}^\circ)
+
\epsilon_{\mathrm{cal}}.
\]
Combining the two inequalities gives
\[
V_{\mathrm{ex}}(S_{\mathrm{on}})
\le
V_{\mathrm{ex}}(S_{\mathrm{myo}}^\circ)
+
2L_\psi E_f
+
2\epsilon_{\mathrm{cal}}.
\]
If \(\epsilon_{f,j}\le\epsilon_f\) for all \(j\), then
\(E_f\le K\epsilon_f\), which yields the uniform-error result.
\end{proof}

Theorem~\ref{thm:online-guarantee} says that Online Calibrated-WA approaches a
one-step clairvoyant overlap scheduler when the forecast error is small. The
comparator is not the full offline optimum; it is the best immediate decision
available at each online decision state.

\section{Claim Hierarchy and Artifact Checklist}

Table~\ref{tab:claim_hierarchy} summarizes the claim boundary used throughout the paper.
The hierarchy is intentionally strict: passing one tier motivates the next experiment, but does not license the next claim.

\begin{table}[H]
\centering
\small
\begin{tabular}{@{}p{0.16\linewidth}p{0.25\linewidth}p{0.27\linewidth}p{0.22\linewidth}@{}}
\toprule
\textbf{Tier} & \textbf{Evidence required} & \textbf{This paper} & \textbf{Valid claim} \\
\midrule
Stress harness & Controlled pressure replay with fixed and deferral baselines & WA changes the SLO/progress frontier & Scheduling can matter in a repeatable benchmark \\
Ordinary sidecar & Real p95 effect-size transfer & Qwen2.5 vLLM H/L p95 is 1.002--1.005 & No production SLO claim from the harness alone \\
Targeted sidecar & Mechanism-specific transfer & Pressure-heavy and sync-heavy regimes transfer differently & Calibrate the mechanism before policy replay \\
Steady replay & Matched policy traces and matched deferral & WA beats fixed/gate, ties random & WA controls sync windows, but random is strong \\
Bursty offline replay & No-sync load match and random-draw envelope & Cal-WA is best tested policy; beats WA and random mean & Calibrated-WA wins in this bursty replay \\
Online controller & Decision-time signal only & Online Cal-WA partially recovers the signal but ties random & Online controller remains incomplete \\
Online forecaster (proxy) & Decision-time EWMA forecast, overlap-proxy scoring & Beats matched random 8/10 seeds ($p=0.020$) & Mechanism identified; precursor to sidecar replay \\
Online forecaster (sidecar) & Same schedules through real vLLM on the same trace & Beats matched random 8/10 seeds ($p=0.021$), 5.09\% vs 6.54\% & Online forecasting beats random in real-sidecar replay; one trace family \\
\bottomrule
\end{tabular}
\caption{Claim hierarchy. Each row states what the evidence supports and what it does not support.}
\label{tab:claim_hierarchy}
\end{table}

The anonymized artifact is available at \url{https://anonymous.4open.science/r/workload-aware-diloco-05FB/README.md}.
It is designed to let readers audit the protocol even when they cannot reproduce 8xB200 or MI355X runs.
It contains the WA-DiLoCo implementation, stress-harness and sidecar-replay scripts, Slurm launchers, policy/config files, aggregate summaries for Stages 10, 11, 13, and 14, Stage 12 replay/schedule scripts, Stage 18 forecaster replay and Stage 19 sidecar-replay summaries, plotting inputs for Figure~\ref{fig:protocol}, and anonymized per-seed CSV/JSON summaries.
We include scripts that regenerate the calibration tables, Stage 11/13 policy summaries, significance tests, and the MI355X/RCCL fabric-smoke table from the released summaries; the Stage 12 scripts document how the bursty replay and random-draw envelopes were generated.
Large model checkpoints, raw serving logs with prompt text, and cluster-specific paths are excluded; the released summaries preserve seed, policy, sync-window, latency, SLO, and throughput fields needed to reproduce the paper tables.

\section{Implementation and Evaluation Details}
\label{app:setup}

\paragraph{WA-DiLoCo configuration.}
The main WA-DiLoCo policy uses min-$H=128$, max-$H=512$, and threshold $0.0$. Loss-progress and token-contribution weights are both $0.5$; network pressure is $0.5$, serving pressure is $0.75$, checkpoint pressure is $0.25$, and staleness is $0.25$. Cap scales are $0.05$ for loss reduction, $10^6$ tokens, and $1800$ seconds of staleness. Serving pressure receives the largest fleet-side weight because the target setting is co-located training and latency-sensitive serving. Network and checkpoint terms are included as shared-fleet signals, but in this study they remain harness signals unless separately calibrated.

\paragraph{Hardware and training.}
Experiments run on one 8$\times$B200 DGX node with one learner per GPU. We continue HuggingFaceTB/SmolLM2-135M on HuggingFaceFW/fineweb-edu using bf16, sequence length 1024, per-device batch size 2, AdamW learning rate $2\cdot10^{-5}$, weight decay 0.1, and gradient clipping 1.0. Held-out evaluation uses WikiText-2. The main frontier uses five seeds, the guarded ablation campaign contains 79 training runs, and the SmolLM2-1.7B scale check uses three seeds.

\paragraph{Serving sidecar and replay.}
Real-serving calibration uses vLLM on the same B200 environment with Qwen2.5-0.5B, 1.5B, 3B, and 7B sidecars; short, mixed, and long prompts; concurrency 1--8; and targeted pressure-heavy or sync-heavy stress. Final policy replay sends fixed-$H$, gate, matched-deferral, and WA traces through Qwen2.5-1.5B-Instruct vLLM sidecars for 600 seconds across five seeds. Bursty replay alternates quiet and busy windows, includes a no-sync load-matched baseline, runs five matched-random draws per seed, and reports sync-excess SLO.

\paragraph{Detailed calibration and policy replay tables.}
The following tables provide the full calibration and Stage 11 replay results summarized in the main text.

\begin{table}[H]
\centering
\small
\begin{tabular}{llrrrrr}
\toprule
\textbf{Run} & \textbf{Condition} & \textbf{Req.} & \textbf{Real H/L p95} & \textbf{Corr.} & \textbf{Sync ratio} & \textbf{Sim H/L p95} \\
\midrule
0.5B baseline & c=1, short & 12,740 & 1.0047 & 0.046 & 1.0008 & 1.7179 \\
0.5B high-conc. & c=8, mixed & 25,664 & 1.0020 & 0.760 & 1.0002 & 1.7191 \\
1.5B medium & c=4, mixed & 11,004 & 1.0031 & 0.076 & 1.0006 & 1.7175 \\
1.5B high-conc. & c=8, long & 14,688 & 1.0034 & 0.773 & 1.0008 & 1.7204 \\
3B high-conc. & c=8, long & 9,200 & 1.0034 & 0.781 & 1.0006 & 1.7163 \\
7B high & c=4, long & 4,036 & 1.0046 & 0.775 & 1.0005 & 1.7088 \\
\midrule
1.5B pressure-heavy & c=8, mixed & 10,240 & \textbf{1.5461} & \textbf{0.898} & 0.9998 & 1.7093 \\
1.5B sync-heavy & c=8, mixed & 11,792 & 1.0001 & 0.016 & \textbf{1.5560} & 1.7115 \\
\bottomrule
\end{tabular}
\caption{Real vLLM calibration matrix. Ordinary sidecars show little p95 transfer despite strong synthetic effects; targeted regimes reveal when pressure or synchronization-window interference is real.}
\label{tab:calibration}
\end{table}

\begin{table}[H]
\centering
\small
\begin{tabular}{lrrrrr}
\toprule
\textbf{Policy} & \textbf{Pressure SLO} & \textbf{Pressure p95} & \textbf{Sync SLO} & \textbf{Sync p95} & \textbf{Sync tok/s} \\
\midrule
Fixed-$H$ 256 & $36.67 \pm 1.61$ & $1238.2$ & $19.21 \pm 1.73$ & $1245.9$ & $1242.5$ \\
Fixed-$H$ 512 & $36.82 \pm 0.57$ & $1212.4$ & $11.04 \pm 1.21$ & $1224.8$ & $1301.3$ \\
Pressure gate & $37.21 \pm 1.35$ & $1228.1$ & $13.73 \pm 0.53$ & $1254.5$ & $1280.8$ \\
Gate matched & $\mathbf{35.91 \pm 0.77}$ & $1283.0$ & $11.93 \pm 3.36$ & $1224.5$ & $1300.6$ \\
Random matched & $36.61 \pm 1.04$ & $1225.4$ & $\mathbf{6.38 \pm 0.71}$ & $1173.3$ & $1333.9$ \\
WA-DiLoCo & $37.07 \pm 1.45$ & $1244.1$ & $6.46 \pm 0.35$ & $1241.5$ & $\mathbf{1342.3}$ \\
\bottomrule
\end{tabular}
\caption{Real vLLM policy replay over five seeds: SLO violation rate (\%), p95 latency (ms), and serving throughput in output tokens per second. Pressure-heavy replay is saturated and does not separate policies; sync-heavy replay separates fixed/gate policies from WA-DiLoCo and matched random deferral.}
\label{tab:stage11}
\end{table}

\paragraph{Signals, release, and statistics.}
Serving cost is expected request overlap with the next 8-second synchronization window. Sync-excess SLO subtracts the no-sync run under the same seed and burst trace. We include per-seed policy outcomes, random-draw envelopes, weight and min/max-$H$ sweeps, and scripts/traces for the stress harness, policy replay, and sidecar summaries. A separate MI355X/RCCL smoke test measures outer-sync-sized 1GiB-per-rank all-reduce payloads from 16 to 48 GPUs (Table~\ref{tab:stage14_fabric}). For matched policy replays, we report exact one-sided paired sign-flip tests over seed-level differences; for unpaired stress-harness and ablation comparisons, we report exact one-sided label-permutation tests. Five-seed sign-flip tests bottom out at $p = 1/32 \approx 0.031$, so a five-seed result at $p=0.031$ conveys exactly a 5/5 sweep; the forecaster comparisons therefore use ten seeds, and ten-seed replication of the main replays is straightforward with the released harness.

\begin{table}[H]
\centering
\small
\begin{tabular}{rrrr}
\toprule
\textbf{Nodes} & \textbf{GPUs} & \textbf{1GiB/rank all-reduce} & \textbf{BusBW} \\
\midrule
2 & 16 & $5.621$ ms & $358.2$ GB/s \\
4 & 32 & $5.767$ ms & $360.7$ GB/s \\
6 & 48 & $6.004$ ms & $350.2$ GB/s \\
\bottomrule
\end{tabular}
\caption{MI355X/RCCL fabric smoke test. Outer-sync-sized collectives are measurable on real multi-node fabric, but this is not a full WA-DiLoCo policy replay.}
\label{tab:stage14_fabric}
\end{table}

\paragraph{Additional result details.}
At a 15.96M-token target, WA-DiLoCo reaches 3.54\% SLO-like violations, compared with 19.38\% for fixed-$H$ 256 and 14.29\% for fixed-$H$ 512 ($p=0.004$). Its gap over pressure gate and random deferral is not separable at five seeds, so the harness supports scheduling sensitivity rather than production dominance. Guarded ablations attribute robustness to combining progress and pressure: full WA reaches 2.10\% SLO-like violations, removing serving pressure, network pressure, loss progress, or token progress raises violations to 5.46--6.32\%, pressure-only reaches 6.34\%, and progress-only degrades to 26.64\%. WA-DiLoCo also defers more often than the pressure gate at matched tokens (54.6 versus 41.6, Table~\ref{tab:loss}) with different timing. A SmolLM2-1.7B scale check preserves the qualitative pattern, held-out WikiText-2 loss remains stable across the 135M frontier, and the MI355X/RCCL smoke test shows outer-sync-sized 1GiB-per-rank collectives are measurable on a separate multi-node fabric (Table~\ref{tab:stage14_fabric}), though not as a full policy replay.

\begin{table}[H]
\centering
\small
\begin{tabular}{lrr}
\toprule
\textbf{Variant} & \textbf{SLO viol.} & \textbf{$p_<$} \\
\midrule
Full WA & $\mathbf{2.10 \pm 2.88}$ & -- \\
No SLO pressure & $5.46 \pm 0.85$ & .032 \\
No network pressure & $5.66 \pm 0.91$ & .040 \\
No loss progress & $5.60 \pm 1.06$ & .044 \\
No token progress & $6.32 \pm 1.11$ & .012 \\
Pressure only & $6.34 \pm 1.00$ & .012 \\
Progress only & $26.64 \pm 6.96$ & .004 \\
Random matched & $4.99 \pm 3.04$ & .071 \\
Gate matched & $5.06 \pm 2.31$ & .040 \\
\bottomrule
\end{tabular}
\caption{Guarded token-matched ablations over five seeds. $p_<$ tests lower WA SLO-like violations.}
\label{tab:guarded_ablation}
\end{table}

\begin{table}[H]
\centering
\small
\begin{tabular}{lrrrr}
\toprule
\textbf{Policy} & \textbf{Tokens M} & \textbf{Train loss} & \textbf{SLO viol.} & \textbf{Defers} \\
\midrule
Fixed-$H$ 256 & $16.136 \pm 0.129$ & $2.649 \pm 0.025$ & $19.381 \pm 2.016$ & $0.000 \pm 0.000$ \\
Fixed-$H$ 512 & $16.136 \pm 0.129$ & $2.678 \pm 0.062$ & $14.293 \pm 0.361$ & $0.000 \pm 0.000$ \\
Pressure gate & $16.136 \pm 0.129$ & $2.661 \pm 0.019$ & $3.932 \pm 0.615$ & $41.600 \pm 8.764$ \\
Random deferral & $16.136 \pm 0.129$ & $2.661 \pm 0.019$ & $4.437 \pm 0.413$ & $38.600 \pm 5.273$ \\
WA-DiLoCo & $16.136 \pm 0.129$ & $2.661 \pm 0.019$ & $3.544 \pm 1.993$ & $54.600 \pm 5.413$ \\
\bottomrule
\end{tabular}
\caption{Token-matched stress-harness frontier with train loss over five seeds. Synchronization policy does not measurably affect short-run train loss at matched tokens, isolating serving SLO as the operative axis.}
\label{tab:loss}
\end{table}

\paragraph{Forecast-augmented online replay.}
The forecaster observes the recent request process with an exponentially weighted moving average and chooses among feasible synchronization windows using the calibrated overlap cost, under the same min/max-$H$ budget accounting. Replay uses one fixed non-periodic request trace over ten seeds; sync-excess SLO is estimated from measured window overlap via the bursty sidecar calibration rather than a live sidecar. Matched random reaches $7.64 \pm 1.78$\%, the forecaster $6.00 \pm 1.36$\%, and the offline oracle $2.71$\%; the forecaster reduces overlapped requests per run from $344 \pm 80$ to $270 \pm 61$ and busy-window synchronization fraction from $0.242$ to $0.125$. The exact one-sided paired sign-flip test over seed differences gives $p=0.020$. Table~\ref{tab:ewma_forecaster} summarizes the replay. The forecaster was specified from the offline-online diagnosis before either replay was run, and it is the only forecast-based online variant we evaluated; the proxy and sidecar replays were each run once with pre-fixed seeds. Separately explored online cost-gate variants ran only on a harness proxy, produced no usable evidence, and are not reported as evidence.

\begin{table}[H]
\centering
\small
\begin{tabular}{@{}lrrr@{}}
\toprule
\textbf{Policy} & \textbf{Overlap} & \textbf{Busy frac.} & \textbf{Excess SLO} \\
\midrule
Random matched & $344.0 \pm 80.4$ & $0.242 \pm 0.135$ & $7.637 \pm 1.784$ \\
EWMA forecaster & $270.2 \pm 61.2$ & $0.125 \pm 0.112$ & $5.998 \pm 1.359$ \\
Offline oracle & $122.0 \pm 0.0$ & $0.000 \pm 0.000$ & $2.708 \pm 0.000$ \\
\bottomrule
\end{tabular}
\caption{Replay-only online forecasting check on one non-periodic serving trace over 10 seeds. The SLO column is an overlap-calibrated proxy from bursty vLLM replay, not a new vLLM sidecar measurement. Lower is better.}
\label{tab:ewma_forecaster}
\end{table}

\paragraph{Real-sidecar forecaster replay.}
The same schedules are then replayed through real Qwen2.5-1.5B-Instruct vLLM sidecars on the same non-periodic trace for 600 seconds per seed, with eight synchronization events per run and roughly 4,200--4,400 requests served per run. The exact one-sided paired sign-flip test over the ten seed differences gives $p=0.021$. The accompanying no-sync runs on this trace are not load-matched and are excluded from sync-excess accounting.

\begin{table}[H]
\centering
\small
\begin{tabular}{lrrrr}
\toprule
\textbf{Policy} & \textbf{SLO viol.} & \textbf{p95 (ms)} &  \textbf{Sync-active reqs} & \textbf{Sync events} \\
\midrule
Random matched & $6.544 \pm 1.601$ & $1120.3 \pm 153.7$ & $286.0 \pm 66.5$ & $8$ \\
EWMA forecaster & $5.087 \pm 1.190$ & $1025.1 \pm 216.4$ & $219.8 \pm 51.7$ & $8$ \\
\bottomrule
\end{tabular}
\caption{Real vLLM sidecar replay of the forecaster and matched-random schedules over ten seeds on the non-periodic trace. Lower is better; both policies use identical synchronization budgets.}
\label{tab:stage19}
\end{table}

\end{document}

%% file: fig_intuition_incl.tex
\begin{tikzpicture}[x=0.80cm,y=1.0cm,font=\scriptsize,line cap=round]
\def\w{0.72} 

\def\loadcoords{(0,0.07) (0.45,0.08) (0.75,0.30) (1.05,0.55) (1.45,0.50)
 (1.80,0.40) (2.10,0.12) (2.70,0.08) (3.15,0.28) (3.55,0.55) (3.95,0.45)
 (4.30,0.12) (5.10,0.06) (5.90,0.08) (6.35,0.32) (6.90,0.55) (7.45,0.46)
 (8.05,0.54) (8.55,0.30) (9.05,0.10) (10,0.06)}

\newcommand{\loadrow}[3]{
  \begin{scope}[yshift=#1cm]
    \fill[ghostgray!40] plot[smooth,tension=0.7] coordinates {\loadcoords}
      -- (10,0) -- (0,0) -- cycle;
    \draw[pressline!85,line width=0.5pt]
      plot[smooth,tension=0.7] coordinates {\loadcoords};
    \draw[darktext!60,line width=0.4pt] (0,0) -- (10,0);
    \node[anchor=west,darktext,inner sep=1pt] at (0,0.95) {#2};
    \node[anchor=east,inner sep=1pt] at (10,0.95) {#3};
  \end{scope}}

\newcommand{\syncwin}[3]{
  \fill[#3,opacity=0.38] ({#2-\w/2},#1) rectangle ({#2+\w/2},{#1+0.62});
  \draw[#3,line width=0.6pt] ({#2-\w/2},#1) rectangle ({#2+\w/2},{#1+0.62});}

\loadrow{3.6}{\textbf{Fixed-$H$} (DiLoCo default)}
        {\color{slored}\textbf{3/3 in bursts}}
\loadrow{1.8}{\textbf{Matched random} (control, same budget)}
        {\color{slored}\textbf{2/3 this draw}}
\loadrow{0}{\textbf{Workload-aware + forecast} (ours)}
        {\color{goodgreen}\textbf{0/3}}

\syncwin{3.6}{1.20}{slored}
\syncwin{3.6}{3.85}{slored}
\syncwin{3.6}{6.50}{slored}
\syncwin{1.8}{2.50}{goodgreen}
\syncwin{1.8}{3.60}{slored}
\syncwin{1.8}{8.00}{slored}
\syncwin{0}{2.55}{goodgreen}
\syncwin{0}{5.35}{goodgreen}
\syncwin{0}{9.45}{goodgreen}
\draw[slored,dashed,line width=0.6pt]
  ({3.70-\w/2},0) rectangle ({3.70+\w/2},0.62);
\draw[-{Stealth},syncblue,line width=0.8pt]
  (4.10,0.36) to[bend left=22] (4.95,0.36);
\node[syncblue,anchor=west,inner sep=1pt] at (5.62,0.70) {deferred past burst};

\node[pressline,anchor=south,inner sep=1pt] at (3.10,4.78) {request bursts};
\draw[-{Stealth},pressline,line width=0.5pt] (3.25,4.74) -- (3.33,4.05);
\node[darktext,anchor=south,inner sep=1pt] at (6.85,4.78) {outer merge ($\sim$8\,s)};
\draw[-{Stealth},darktext,line width=0.5pt] (6.70,4.74) -- (6.55,4.30);

\draw[-{Stealth},darktext!70,line width=0.5pt] (0,-0.24) -- (10,-0.24);
\node[darktext!70,anchor=north east,inner sep=1pt] at (10,-0.30) {time};
\fill[slored,opacity=0.38] (0,-0.82) rectangle (0.30,-0.62);
\draw[slored,line width=0.5pt] (0,-0.82) rectangle (0.30,-0.62);
\node[anchor=west,darktext,inner sep=1pt] at (0.38,-0.72)
  {merge in burst $\rightarrow$ SLO violations};
\fill[goodgreen,opacity=0.38] (6.10,-0.82) rectangle (6.40,-0.62);
\draw[goodgreen,line width=0.5pt] (6.10,-0.82) rectangle (6.40,-0.62);
\node[anchor=west,darktext,inner sep=1pt] at (6.48,-0.72)
  {merge in valley $\rightarrow$ cheap};
\end{tikzpicture}

%% file: fig1_incl.tex
\definecolor{waBlue}{HTML}{0072B2}
\definecolor{calPurple}{HTML}{CC79A7}
\definecolor{randGreen}{HTML}{009E73}
\definecolor{gateOrange}{HTML}{E69F00}
\definecolor{fixedGray}{HTML}{6E6E6E}
\definecolor{lightGray}{HTML}{BFBFBF}
\definecolor{warnRed}{HTML}{D55E00}
\definecolor{ink}{HTML}{1A1A1A}
\definecolor{panelBg}{HTML}{FAFAFA}
\definecolor{panelEdge}{HTML}{D9D9D9}
\definecolor{slate}{HTML}{4D6B82}

\colorlet{accentA}{slate}
\colorlet{accentB}{warnRed}
\colorlet{accentC}{waBlue}
\colorlet{accentD}{calPurple}

\def\Afixed{19.4}\def\Agate{3.9}\def\Arand{4.4}\def\Awa{3.5}
\def\Bsyn{1.71}
\def\Cfixed{19.2}\def\CfixedB{11.0}\def\Cgate{13.7}\def\Crand{6.38}\def\Cwa{6.46}
\def\Dfixed{19.0}\def\Dgate{12.7}\def\Drand{6.3}\def\Dwa{6.9}\def\Dcal{4.4}

\newcommand{\slobar}[5]{%
  \pgfmathsetmacro{\bh}{#3*0.13}%
  \pgfmathsetmacro{\cx}{#1+#2/2}%
  \fill[#4] (#1,0) rectangle ++(#2,\bh);%
  \node[vlab] at (\cx,\bh+0.18) {#3};%
  \node[xlab] at (\cx,-0.18) {#5};%
}

\resizebox{\textwidth}{!}{%
\begin{tikzpicture}[
  font=\sffamily,
  >=Latex,
  panel/.style={rounded corners=4pt, draw=panelEdge, line width=0.5pt,
                fill=panelBg, minimum width=4.7cm, minimum height=6.58cm,
                anchor=north west},
  ptitle/.style={font=\footnotesize\bfseries, text=ink, anchor=west},
  psub/.style={font=\scriptsize, text=black!55, anchor=west},
  metric/.style={font=\scriptsize, text=black!70, anchor=west},
  tick/.style={font=\scriptsize, text=black!55, anchor=east},
  vlab/.style={font=\scriptsize, text=ink, anchor=south},
  xlab/.style={font=\scriptsize, text=black!75, anchor=north},
  axis/.style={black!45, line width=0.4pt},
  grid/.style={black!10, line width=0.3pt, dash pattern=on 1pt off 1.4pt},
  callout/.style={rounded corners=2.5pt, fill=white, draw=panelEdge,
                  line width=0.4pt, font=\scriptsize, align=center,
                  text=black!72, inner sep=3pt, text width=3.4cm},
  badge/.style={circle, minimum size=0.42cm, inner sep=0pt,
                font=\scriptsize\bfseries, text=white,
                text height=1.6ex, text depth=0.25ex},
]

\def\PA{0}    \def\PB{5.2}  \def\PC{10.4} \def\PD{15.6}
\def\PW{4.7}  \def\PH{7.0}
\pgfmathsetmacro{\midY}{-3.6}
\pgfmathsetmacro{\baseY}{-5.2}
\path[use as bounding box] (-0.12,-7.08) rectangle (20.35,1.03);

\node[anchor=west, font=\large\bfseries, text=ink] at (\PA,1.02)
  {Calibrated Workload-Aware DiLoCo};
\node[anchor=west, font=\normalsize, text=black!55] at (\PA,0.57) {};
\draw[panelEdge, line width=0.6pt] (\PA,0.27) -- (20.3,0.27);

\foreach \px/\acc in {\PA/accentA,\PB/accentB,\PC/accentC,\PD/accentD}{
  \node[panel] at (\px,0) {};
  \fill[\acc, rounded corners=4pt] (\px,0) rectangle ++(\PW,-0.14);
}

\draw[->, line width=0.8pt, draw=black!45]
  ($(\PA,\midY)+(\PW+0.05,0)$) -- ($(\PB,\midY)+(-0.05,0)$);
\draw[->, line width=0.8pt, draw=black!45]
  ($(\PB,\midY)+(\PW+0.05,0)$) -- ($(\PC,\midY)+(-0.05,0)$);
\draw[->, line width=0.8pt, draw=black!45]
  ($(\PC,\midY)+(\PW+0.05,0)$) -- ($(\PD,\midY)+(-0.05,0)$);

\begin{scope}
  \node[badge, fill=accentA, anchor=west] (p1badge) at (\PA+0.18,-0.58) {1};
  \node[ptitle, right=6pt of p1badge] (p1title) {Synthetic harness};
  \node[psub, below=3pt of p1title.west, yshift=-6pt, anchor=west] {controlled SLO-like benchmark};
  \node[metric, below=3pt of p1title.west, yshift=-18pt, anchor=west] {SLO-like violations (\%)};
\end{scope}
\begin{scope}[shift={(\PA+0.78,\baseY)}]
  \draw[grid] (0,1.3) -- (3.7,1.3);  \draw[grid] (0,2.6) -- (3.7,2.6);
  \draw[axis] (0,0) -- (3.7,0);       \draw[axis] (0,0) -- (0,2.78);
  \node[tick] at (-0.06,0) {0}; \node[tick] at (-0.06,1.3) {10}; \node[tick] at (-0.06,2.6) {20};
  \slobar{0.30}{0.42}{\Afixed}{fixedGray}{F256}
  \slobar{1.18}{0.42}{\Agate}{gateOrange}{Gate}
  \slobar{2.06}{0.42}{\Arand}{randGreen}{Rand}
  \slobar{2.94}{0.42}{\Awa}{waBlue}{WA}
\end{scope}
\node[callout, anchor=north] at (\PA+\PW/2,-5.92)
  {A stress-test win; not deployment.};

\begin{scope}
  \node[badge, fill=accentB, anchor=west] (p2badge) at (\PB+0.18,-0.58) {2};
  \node[ptitle, right=6pt of p2badge] (p2title) {Ordinary vLLM};
  \node[psub, below=3pt of p2title.west, yshift=-6pt, anchor=west] {p95 effect-size gate fails};
  \node[metric, below=3pt of p2title.west, yshift=-18pt, anchor=west] {high/low p95 ratio};
  \node[metric, text=black!45, below=2pt of p2title.west, yshift=-32pt, anchor=west] {1.0$\times$ = no effect};
\end{scope}
\begin{scope}[shift={(\PB+0.92,\baseY)}]
  \draw[grid] (0,1.3) -- (3.5,1.3);  \draw[grid] (0,2.6) -- (3.5,2.6);
  \draw[axis] (0,0) -- (3.5,0);       \draw[axis] (0,0) -- (0,2.78);
  \node[tick] at (-0.06,0) {1.0$\times$}; \node[tick] at (-0.06,1.3) {1.4$\times$}; \node[tick] at (-0.06,2.6) {1.8$\times$};
  \draw[warnRed, line width=0.6pt, dash pattern=on 2pt off 1.6pt] (0,0.325) -- (3.5,0.325);
  \node[font=\scriptsize, text=warnRed, anchor=east] at (3.46,0.47) {1.1$\times$ gate};
  \pgfmathsetmacro{\synh}{(\Bsyn-1)*3.25}
  \fill[warnRed] (0.55,0) rectangle ++(0.60,\synh);
  \node[vlab] at (0.85,\synh+0.18) {\Bsyn$\times$};
  \node[xlab] at (0.85,-0.18) {synthetic};
  \fill[lightGray] (2.05,0) rectangle ++(0.60,0.04);
  \draw[black!45, line width=0.6pt] (2.05,0.04) -- ++(0.60,0);
  \node[vlab, text=black!70] at (2.35,0.82) {1.002--1.005$\times$};
  \draw[black!30, line width=0.3pt] (2.35,0.76) -- (2.35,0.10);
  \node[xlab] at (2.35,-0.18) {real sidecar};
\end{scope}
\node[callout, anchor=north] at (\PB+\PW/2,-5.92)
  {Effect size, not correlation, gates it.};

\begin{scope}
  \node[badge, fill=accentC, anchor=west] (p3badge) at (\PC+0.18,-0.58) {3};
  \node[ptitle, right=6pt of p3badge] (p3title) {Steady sync-heavy};
  \node[psub, below=3pt of p3title.west, yshift=-6pt, anchor=west] {Stage 11 real-vLLM replay};
  \node[metric, below=3pt of p3title.west, yshift=-18pt, anchor=west] {SLO violations (\%)};
  \node[metric, text=black!45, below=2pt of p3title.west, yshift=-32pt, anchor=west] {1000 ms threshold};
\end{scope}
\begin{scope}[shift={(\PC+0.78,\baseY)}]
  \draw[grid] (0,1.3) -- (3.7,1.3);  \draw[grid] (0,2.6) -- (3.7,2.6);
  \draw[axis] (0,0) -- (3.7,0);       \draw[axis] (0,0) -- (0,2.78);
  \node[tick] at (-0.06,0) {0}; \node[tick] at (-0.06,1.3) {10}; \node[tick] at (-0.06,2.6) {20};
  \slobar{0.18}{0.36}{\Cfixed}{fixedGray}{F256}
  \slobar{0.86}{0.36}{\CfixedB}{lightGray}{F512}
  \slobar{1.54}{0.36}{\Cgate}{gateOrange}{Gate}
  \slobar{2.22}{0.36}{\Crand}{randGreen}{Rand}
  \slobar{2.90}{0.36}{\Cwa}{waBlue}{WA}
  \draw[black!40, line width=0.5pt]
       (2.40,1.34) -- (2.40,1.46) -- (3.08,1.46) -- (3.08,1.34);
  \node[font=\scriptsize, text=black!55] at (2.74,1.62) {$\approx$};
\end{scope}
\node[callout, anchor=north] at (\PC+\PW/2,-5.92)
  {Beats fixed-$H$/gates; random stays competitive.};

\begin{scope}
  \node[badge, fill=accentD, anchor=west] (p4badge) at (\PD+0.18,-0.58) {4};
  \node[ptitle, right=6pt of p4badge] (p4title) {Bursty calibrated};
  \node[psub, below=3pt of p4title.west, yshift=-6pt, anchor=west] {Stage 12 real-vLLM replay};
  \node[metric, below=3pt of p4title.west, yshift=-18pt, anchor=west] {sync-excess SLO (\%)};
  \node[metric, text=black!45, below=2pt of p4title.west, yshift=-32pt, anchor=west] {vs. no-sync load match};
\end{scope}
\begin{scope}[shift={(\PD+0.78,\baseY)}]
  \draw[grid] (0,1.3) -- (3.7,1.3);  \draw[grid] (0,2.6) -- (3.7,2.6);
  \draw[axis] (0,0) -- (3.7,0);       \draw[axis] (0,0) -- (0,2.78);
  \node[tick] at (-0.06,0) {0}; \node[tick] at (-0.06,1.3) {10}; \node[tick] at (-0.06,2.6) {20};
  \slobar{0.18}{0.36}{\Dfixed}{fixedGray}{F256}
  \slobar{0.86}{0.36}{\Dgate}{gateOrange}{Gate}
  \slobar{1.54}{0.36}{\Drand}{randGreen}{Rand}
  \slobar{2.22}{0.36}{\Dwa}{waBlue}{WA}
  \slobar{2.90}{0.36}{\Dcal}{calPurple}{Cal}
\end{scope}
\node[callout, anchor=north] at (\PD+\PW/2,-5.92)
  {Calibrated WA exploits structure random lacks.};

\begin{scope}[shift={(\PA,-6.92)}]
  \node[anchor=west, font=\scriptsize\bfseries, text=ink] at (0,0) {Policies};
  \pgfmathsetmacro{\legLeft}{2.10}  
  \pgfmathsetmacro{\legRight}{18.40} 
  \pgfmathsetmacro{\legStep}{(\legRight-\legLeft)/4.0}
  \foreach \i/\c/\t in {0/fixedGray/{fixed-$H$}, 1/gateOrange/{pressure gate}, 2/randGreen/{matched random}, 3/waBlue/{WA-DiLoCo}, 4/calPurple/{calibrated WA}}{
    \pgfmathsetmacro{\lx}{\legLeft+\i*\legStep}
    \fill[\c, rounded corners=1pt] (\lx,-0.08) rectangle ++(0.24,0.16);
    \node[anchor=west, font=\scriptsize, text=black!75] at (\lx+0.32,0) {\t};
  }
\end{scope}

\end{tikzpicture}}%

%% file: custom.bib
@misc{douillard2023diloco,
  title = {{DiLoCo}: Distributed Low-Communication Training of Language Models},
  author = {Douillard, Arthur and Feng, Qixuan and Rusu, Andrei A. and Chhaparia, Rachita and Donchev, Yani and Kuncoro, Adhiguna and Ranzato, Marc'Aurelio and Szlam, Arthur and Shen, Jiajun},
  year = {2023},
  eprint = {2311.08105},
  archivePrefix = {arXiv},
  primaryClass = {cs.LG},
  doi = {10.48550/arXiv.2311.08105},
  url = {https://arxiv.org/abs/2311.08105}
}

@misc{douillard2025streaming,
  title = {Streaming {DiLoCo} with Overlapping Communication: Towards a Distributed Free Lunch},
  author = {Douillard, Arthur and Donchev, Yanislav and Rush, Keith and Kale, Satyen and Charles, Zachary and Garrett, Zachary and Teston, Gabriel and Lacey, Dave and McIlroy, Ross and Shen, Jiajun and Ram{\'e}, Alexandre and Szlam, Arthur and Ranzato, Marc'Aurelio and Barham, Paul},
  year = {2025},
  eprint = {2501.18512},
  archivePrefix = {arXiv},
  primaryClass = {cs.CL},
  doi = {10.48550/arXiv.2501.18512},
  url = {https://arxiv.org/abs/2501.18512}
}

@misc{douillard2026decoupled,
  title = {Decoupled {DiLoCo} for Resilient Distributed Pre-training},
  author = {Douillard, Arthur and Rush, Keith and Donchev, Yani and Charles, Zachary and Fallen, Nova and Dubey, Ayush and Gog, Ionel and Dean, Josef and Woodworth, Blake and Garrett, Zachary and Keating, Nate and Bishop, Jenny and Prior, Henry and Yvinec, Edouard and Szlam, Arthur and Ranzato, Marc'Aurelio and Dean, Jeff},
  year = {2026},
  eprint = {2604.21428},
  archivePrefix = {arXiv},
  primaryClass = {cs.CL},
  doi = {10.48550/arXiv.2604.21428},
  url = {https://arxiv.org/abs/2604.21428}
}

@misc{kolehmainen2025noloco,
  title = {{NoLoCo}: No-all-reduce Low Communication Training Method for Large Models},
  author = {Kolehmainen, Jari and Blagoev, Nikolay and Donaghy, John and Ersoy, O{\u{g}}uzhan and Nies, Christopher},
  year = {2025},
  eprint = {2506.10911},
  archivePrefix = {arXiv},
  primaryClass = {cs.LG},
  doi = {10.48550/arXiv.2506.10911},
  url = {https://arxiv.org/abs/2506.10911}
}

@misc{sarfi2025sparseloco,
  title = {Communication Efficient {LLM} Pre-training with {SparseLoCo}},
  author = {Sarfi, Amir and Th{\'e}rien, Benjamin and Lidin, Joel and Belilovsky, Eugene},
  year = {2025},
  eprint = {2508.15706},
  archivePrefix = {arXiv},
  primaryClass = {cs.LG},
  doi = {10.48550/arXiv.2508.15706},
  url = {https://arxiv.org/abs/2508.15706}
}

@misc{du2024opendiloco,
  title = {{OpenDiLoCo}: An Open-Source Framework for Globally Distributed Low-Communication Training},
  author = {Jaghouar, Sami and Ong, Jack Min and Hagemann, Johannes},
  year = {2024},
  eprint = {2407.07852},
  archivePrefix = {arXiv},
  primaryClass = {cs.LG},
  doi = {10.48550/arXiv.2407.07852},
  url = {https://arxiv.org/abs/2407.07852}
}

@inproceedings{koneputugodage2026factored,
  title = {Factored Gossip {DiLoCo}: Reducing Blocking Communication within {DiLoCo}},
  author = {Koneputugodage, Chamin Hewa and Ajanthan, Thalaiyasingam and Ramasinghe, Sameera and Dolatabadi, Hadi Mohaghegh and Siriwardhana, Shamane and Avraham, Gil and Shevchenko, Violetta and Pajak, Karol and Snewin, James and Long, Alexander},
  booktitle = {Proceedings of the 43rd International Conference on Machine Learning (ICML)},
  year = {2026},
  note = {To appear},
  url = {https://icml.cc/virtual/2026/poster/66683}
}

@inproceedings{kwon2023pagedattention,
  title = {Efficient Memory Management for Large Language Model Serving with {PagedAttention}},
  author = {Kwon, Woosuk and Li, Zhuohan and Zhuang, Siyuan and Sheng, Ying and Zheng, Lianmin and Yu, Cody Hao and Gonzalez, Joseph E. and Zhang, Hao and Stoica, Ion},
  booktitle = {Proceedings of the 29th Symposium on Operating Systems Principles},
  year = {2023},
  pages = {611--626},
  doi = {10.1145/3600006.3613165},
  url = {https://doi.org/10.1145/3600006.3613165}
}

@misc{huang2026collm,
  title = {{CoLLM}: Continuous Adaptation for {SLO}-Aware {LLM} Serving on Shared {GPU} Clusters},
  author = {Huang, Shaoyuan and Zhao, Yunfeng and Yan, Na and Zhang, Tiancheng and Wang, Xiaokai and Wang, Xiaofei and Wang, Wenyu and Deng, Yansha},
  year = {2026},
  eprint = {2604.16400},
  archivePrefix = {arXiv},
  primaryClass = {cs.DC},
  doi = {10.48550/arXiv.2604.16400},
  url = {https://arxiv.org/abs/2604.16400}
}

@misc{li2025lemix,
  title = {{LeMix}: Unified Scheduling for {LLM} Training and Inference on Multi-{GPU} Systems},
  author = {Li, Yufei and Li, Zexin and Zhu, Yinglun and Liu, Cong},
  year = {2025},
  eprint = {2507.21276},
  archivePrefix = {arXiv},
  primaryClass = {cs.AI},
  doi = {10.48550/arXiv.2507.21276},
  note = {46th IEEE Real-Time Systems Symposium (RTSS 2025)},
  url = {https://arxiv.org/abs/2507.21276}
}

@misc{li2025mace,
  title = {{MACE}: A Hybrid {LLM} Serving System with Colocated {SLO}-aware Continuous Retraining Alignment},
  author = {Li, Yufei and Fu, Yu and Dong, Yue and Liu, Cong},
  year = {2025},
  eprint = {2510.03283},
  archivePrefix = {arXiv},
  primaryClass = {cs.AI},
  doi = {10.48550/arXiv.2510.03283},
  url = {https://arxiv.org/abs/2510.03283}
}

@inproceedings{oliaro2026flexllm,
  title = {{FlexLLM}: Token-Level Co-Serving of {LLM} Inference and Finetuning with {SLO} Guarantees},
  author = {Oliaro, Gabriele and Miao, Xupeng and Cheng, Xinhao and Kada, Vineeth and Wu, Mengdi and Gao, Ruohan and Huang, Yingyi and Delacourt, Remi and Yang, April and Wang, Yingcheng and Unger, Colin and Jia, Zhihao},
  booktitle = {23rd USENIX Symposium on Networked Systems Design and Implementation (NSDI 26)},
  year = {2026},
  pages = {1359--1379},
  isbn = {978-1-939133-54-0},
  address = {Renton, WA},
  publisher = {USENIX Association},
  month = may,
  url = {https://www.usenix.org/conference/nsdi26/presentation/oliaro}
}

@inproceedings{chen2025mudi,
  title = {Multiplexing Dynamic Deep Learning Workloads with {SLO}-awareness in {GPU} Clusters},
  author = {Chen, Wenyan and Lu, Chengzhi and Xu, Huanle and Ye, Kejiang and Xu, Chengzhong},
  booktitle = {Proceedings of the Twentieth European Conference on Computer Systems},
  year = {2025},
  pages = {589--604},
  doi = {10.1145/3689031.3696074},
  url = {https://doi.org/10.1145/3689031.3696074}
}

@inproceedings{qiao2020pollux,
  title = {{Pollux}: Co-adaptive Cluster Scheduling for Goodput-Optimized Deep Learning},
  author = {Qiao, Aurick and Choe, Sang Keun and Subramanya, Suhas Jayaram and Neiswanger, Willie and Ho, Qirong and Zhang, Hao and Ganger, Gregory R. and Xing, Eric P.},
  booktitle = {15th USENIX Symposium on Operating Systems Design and Implementation (OSDI 21)},
  year = {2021},
  pages = {1--18},
  isbn = {978-1-939133-22-9},
  publisher = {USENIX Association},
  month = jul,
  url = {https://www.usenix.org/conference/osdi21/presentation/qiao}
}

@inproceedings{xiao2018gandiva,
  title = {Gandiva: Introspective Cluster Scheduling for Deep Learning},
  author = {Xiao, Wencong and Bhardwaj, Romil and Ramjee, Ramachandran and Sivathanu, Muthian and Kwatra, Nipun and Han, Zhenhua and Patel, Pratyush and Peng, Xuan and Zhao, Hanyu and Zhang, Quanlu and Yang, Fan and Zhou, Lidong},
  booktitle = {13th USENIX Symposium on Operating Systems Design and Implementation (OSDI 18)},
  year = {2018},
  pages = {595--610},
  isbn = {978-1-939133-08-3},
  address = {Carlsbad, CA},
  publisher = {USENIX Association},
  month = oct,
  url = {https://www.usenix.org/conference/osdi18/presentation/xiao}
}

@inproceedings{delimitrou2013paragon,
  author    = {Delimitrou, Christina and Kozyrakis, Christos},
  title     = {Paragon: {QoS}-Aware Scheduling for Heterogeneous Datacenters},
  booktitle = {Proceedings of the 18th International Conference on Architectural
               Support for Programming Languages and Operating Systems (ASPLOS)},
  year      = {2013},
  pages     = {77--88}
}

@inproceedings{delimitrou2014quasar,
  author    = {Delimitrou, Christina and Kozyrakis, Christos},
  title     = {Quasar: Resource-Efficient and {QoS}-Aware Cluster Management},
  booktitle = {Proceedings of the 19th International Conference on Architectural
               Support for Programming Languages and Operating Systems (ASPLOS)},
  year      = {2014},
  pages     = {127--144}
}
